\documentclass{easychair}
\providecommand{\keywords}[1]{\textbf{\textit{Index terms---}} #1}
\pdfmapfile{+custom.map}
\newcommand\customfont[1]{{\usefont{T1}{custom}{m}{n} #1 }}

\usepackage{cite}
\usepackage{amsmath,amssymb,amsfonts}
\usepackage{algorithmic}
\usepackage{graphicx}
\usepackage{textcomp}
\usepackage{xcolor}
\usepackage{url}
\usepackage{hyperref}
\usepackage{comment}
\usepackage{subcaption}
\usepackage{mathtools}
\usepackage{calligra}

\usepackage[colorinlistoftodos,prependcaption,textsize=tiny]{todonotes} 
\usepackage{multirow}

\newcommand{\6}{\mathbf }
\newcommand{\hash}{{\sf{Hash}}}
\newcommand{\concat}{{\sf{Concat}}}
\newcommand{\merkle}{{\mathcal{T}}}
\newcommand{\mroot}{{\merkle.\sf{Root}}}

\newcommand{\cmt}{\mathcal{T}'}
\newcommand{\cmtroot}{\cmt.\sf{Root}}

\newcommand{\adversary}{\text{\large \customfont{A}}}
\newcommand{\player}{\text{\large \customfont{P}}}
\newcommand{\oracle}{\text{\large \customfont{O}}}
\newcommand{\code}{\mathcal C}

\usepackage{tikz}
\usetikzlibrary{arrows,automata, positioning}
\usetikzlibrary{calc, chains,
                fit,
                positioning,
                shapes}
\usetikzlibrary{patterns}       

\usepackage{pgfplots}

\usetikzlibrary{fadings,decorations.pathmorphing}

\usetikzlibrary{automata,arrows,positioning,calc}
\definecolor{myblue}{RGB}{0,0,205}
\definecolor{mygreen}{RGB}{80,160,80}
\definecolor{myred}{RGB}{178,34,34}
\definecolor{mygray}{RGB}{211,211,211}

\usepackage[linesnumbered, ruled, boxed]{algorithm2e}

\SetCommentSty{mycommfont}

\usepackage{amsthm}
\newtheorem{Def}{Definition}
\newtheorem{Prob}{Problem}
\newtheorem{Prop}{Proposition}

\newcommand{\mymid}[1]{\hspace{1mm}\left|\hspace{1mm}#1\right.}

\usetikzlibrary{patterns}
\usepackage{scalerel}

\newcommand{\ie}{\textit{i.e.}}

\usepackage{ragged2e}
\usepackage{tabularx}
\newcolumntype{Y}{>{\Centering\arraybackslash}X}

\newcolumntype{C}{>{\centering\arraybackslash}p{0.12\textwidth}}
\newcolumntype{D}{>{\centering\arraybackslash}p{0.275\textwidth}}
\newcolumntype{E}{>{\centering\arraybackslash}p{0.05\textwidth}}
\newcolumntype{F}{>{\centering\arraybackslash}p{0.15\textwidth}}
\newcolumntype{G}{>{\centering\arraybackslash}p{0.115\textwidth}}

\pgfplotsset{compat=1.17}

\begin{document}

\title{Analysis of a blockchain protocol based on LDPC codes}

\author{
Massimo Battaglioni\inst{1}
\and
Paolo Santini \inst{1}
\and
    Giulia Rafaiani\inst{1}
\and
   Franco Chiaraluce\inst{1}
\and
    Marco Baldi\inst{1}
}

\institute{
  Department of Information Engineering, Universit\`{a} Politecnica delle Marche, Ancona, 60131, Italy (e-mail: { \{m.battaglioni, p.santini, g.rafaiani, f.chiaraluce, m.baldi\}@univpm.it}).
  }

\authorrunning{Battaglioni et al.}

\titlerunning{Analysis of a blockchain protocol based on LDPC codes}

\maketitle

\begin{abstract}

In a blockchain Data Availability Attack (DAA), a malicious node publishes a block header but withholds part of the block, which contains invalid transactions. Honest full nodes, which can download and store the full blockchain, are aware that some data are not available but they have no formal way to prove it to light nodes, \ie, nodes that have limited resources and are not able to access the whole blockchain data.
A common solution to counter these attacks exploits linear error correcting codes to encode the block content.
A recent protocol, called SPAR, employs coded Merkle trees and low-density parity-check codes to counter DAAs. In this paper, we show that the protocol is less secure than claimed, owing to a redefinition of the adversarial success probability. As a consequence we show that, for some realistic choices of the parameters, the total amount of  data downloaded by light nodes is larger than that obtainable with competitor solutions.
\end{abstract}


\section{Introduction}
\label{sec:introduction}
A blockchain can be seen as an ordered list of blocks,  each containing a set of transactions occurred among the participants of a peer-to-peer network. The recent discovery of Data Availability Attacks (DAAs) represents a new threat against blockchain security. Since the DAA introduction in \cite{Musta2019}, there has been a growing research interest in finding efficient countermeasures to this type of attacks, possibly leading to new blockchain models with improved scalability and security (e.g., \cite{yu2020coded, Dolececk2020, albassam2019lazyledger, Cao2020}).

In fact, scalability, which is related to the ability of supporting large transaction rates, represents one of the main issues in most existing blockchains \cite{Zhou2020}.
The straightforward solution of increasing the block size raises a series of further concerns. 
In fact, the larger the block size the smaller the number of nodes able to download the full blockchain and, indeed, to participate in the network as \textit{full nodes}, verifying the validity of new blocks and of every contained transaction. More peers would rather participate in the network as \textit{light nodes}, which, due to their limited resources, store only a squeezed version of the blockchain  \cite{Nakamoto} and consequently cannot autonomously verify the validity of transactions.
Light nodes aim at downloading as less data as possible. For instance, they may store only the block headers, which unambiguously identify the content of the blocks. However, in a setting with relatively few full nodes, collusion among them is more probable; this makes light nodes more susceptible to DAAs.
In fact, the aim of a DAA is to make at least one light node accept a block which has not been fully disclosed to the network.
This can happen if and only if honest full nodes are prevented from preparing \textit{fraud proofs}, \textit{i.e,}, demonstrations that the block is invalid \cite{yu2020coded, al2021fraud}.

One of the most promising countermeasures to DAAs consists in encoding the blocks through some error correcting code.
Encoding introduces redundancy and distributes the information of each transaction across all the codeword symbols, so that recovering a small portion of an encoded block may be enough to retrieve the entirety of its contents through decoding.
This strategy, combined with a sampling process in which light nodes ask for fragments of an encoded block and then gossip them to full nodes, ensures that malicious block producers are forced to reveal enough pieces of the invalid block \cite{al2021fraud}.
An alternative to transactions encoding is to change the protocol in such a way that a group of light nodes can collaboratively (among themselves) and autonomously (from full nodes) verify blocks \cite{Cao2020}.  Another option is to decouple the consensus rules from the transaction validity rules \cite{albassam2019lazyledger}. 

In a recent paper \cite{yu2020coded}, Yu et al. proposed SPAR, a blockchain protocol which uses Low-Density Parity-Check (LDPC) codes to counter DAAs; LDPC codes for this specific application have then been studied in \cite{Dolececk2020,Mitra2021}. 
SPAR comes as an improvement of the protocol in \cite{al2021fraud} using two-dimensional Reed-Solomon codes, whose parameters have been optimized in \cite{santini2022optimization}.
The authors of SPAR study the protection against DAAs in case the adversary aims to prevent honest full nodes from successfully decoding the block, which is a strict requirement to settle a proper fraud proof. 
In \cite{yu2020coded}, this situation is investigated assuming the adversary operates by withholding pieces of the encoded block; under a coding theory perspective, this gets modeled as a transmission over an erasure channel.
They conclude that, unless the adversary is able to find stopping sets (which is a NP-hard problem \cite{Krishnan2007}), SPAR guarantees that the success probability of a DAA is sufficiently small even when light nodes download a small amount of data besides the block header. 
As a consequence, SPAR claims improvements in all the relevant metrics \cite[Table 1]{yu2020coded}. 

\paragraph*{Our contribution}

In this paper we study the security of the SPAR protocol.
Namely, we recompute the adversarial success probability with the consideration that deceiving \emph{at least} a single light node is a success for the attacker, which is the same scenario considered in \cite{al2021fraud}.
This yields a sampling cost that is much larger than the expected one, thus penalizing the light nodes participating in the network. Moreover, we show that the total amount of data that light nodes have to download (header size plus sampling cost) is actually larger than that of competing solutions such as \cite{al2021fraud}.

\paragraph*{Paper organization}

The paper is organized as follows. In Section \ref{sec:notation} we describe the notation and some background. In Section \ref{sec:spar} we introduce a general framework to study DAAs. In Section \ref{sec:daa} we provide some numerical results. Finally, in Section \ref{sec:concl} we draw some conclusions.

\section{Notation and background}
\label{sec:notation}
In this section we establish the notation used throughout the paper, and recall some background notions.

\subsection{Mathematical notation}

Given two integers $a$ and $b$, we use $[a , b]$ to indicate the set of integers $x$ such that $a\leq x \leq b$. 
For a set $A$, we use $|A|$ to denote its cardinality.
We denote with $\mathbb F_q$ the finite field with $q$ elements.
Given a vector $\6v$, we use $\mathrm{supp}(\6v)$ to denote its support, \emph{i.e.}, the set containing the positions of its non-zero entries and $w_\mathrm{H}(\6v)$ to denote its Hamming weight, that is, the size of its support. Given an integer $l$ and a set $A$, $A^l$ is the set of vectors of length $l$ taking entries in $A$. Given a matrix $\6M$, $m_{i,j}$ denotes its entry at row $i$ and column $j$,  $\mathbf{M}_{i,:}$ denotes the $i$-th row, and $\mathbf{M}_{:,j}$ denotes the $j$-th column. Given a set $A$, $\6M_{:,A}$ (respectively, $\6M_{A,:}$) represents the matrix formed by the columns (respectively, rows) of $\6M$ indexed by $A$.

We denote by $\concat$ the string concatenation function and by $b(\cdot)$ the binary entropy function. Moreover, we denote by $\hash$ a cryptographic hash function, with codomain $D$. Given some vector $\6a$, we use $\mathcal T(\6a)$ to denote a generic hash tree structure constructed from $\6a$ and using  $\hash$ as underlying function.
The root of the tree is denoted as $\mathcal T.{\sf{Root}}(\6a)$; it generically takes values in $D^t$ and is a one-way function.
With analogous notation, by $\mathcal T.\textsf{Proof}(\6a,i)$ we refer to the proof that the $i$-th entry of $\6a$ is a leaf in the base layer of the tree. 
Notice that, when the hash function $\hash$ is properly chosen, then for any pair of strings $\6a\neq \6a'$ we have $\mathcal T.{\sf{Root}}(\6a)\neq \mathcal T.{\sf{Root}}(\6a')$ and, for any index $i$,
$\mathcal T.{\sf{Proof}}(\6a,i)\neq \mathcal T.{\sf{Proof}}(\6a',i)$ with overwhelming probability (say, not lower than $1-2^{-256}$ for modern hash functions); therefore, for the sake of simplicity, in the following we assume the absence of root and proof collisions. 
%

\subsection{LDPC codes}\label{subsec:ldpccodes}

LDPC codes are a family of linear codes characterized by parity-check matrices having a relatively small number of non-zero entries compared to the number of zeros.
Namely, if an LDPC $\6H\in\mathbb F_q^{r\times n}$ has full rank $r<n$ and row and column weight in the order of $\log(n)$ and $\log(r)$, respectively, then it defines an LDPC code with length $n$ and dimension $k = n-r$.
The associated code is $\code = \left\{\6c\in\mathbb F_q^n\mymid{ \6c\6H^\top = \60}\right\}$, where $^\top$ denotes transposition. 
The rows of the parity-check matrix define the code \emph{parity-check equations}, that is,
\begin{equation}
\label{eq:paritychecks}
\sum_{j = 1}^n c_j h_{i,j} = 0,\hspace{2mm}\forall i \in [1,r],\hspace{2mm}\forall \6c\in\code.
\end{equation}
Equivalently, any code can be represented in terms of a generator matrix $\6G\in\mathbb F_q^{k\times n}$, which forms a basis for $\code$.

In an Erasure Channel (EC), some of the codeword symbols are replaced with the erasure symbol $\epsilon$.  
To this end, we express the action of an EC as $\6c+\6e'$, where $\6c$ is the input sequence and $\6e'\in\{0,\epsilon\}^n$, with $\epsilon$ such that $\epsilon+a = \epsilon$, $\forall a\in\mathbb F_q$. 
A decoding algorithm for the EC aims to obtain a codeword by substituting each erasure with an element from $\mathbb F_q$.
 In the case of LDPC codes, the most common decoder used over the EC is the peeling decoder \cite{Luby}.  
This algorithm works by expressing \eqref{eq:paritychecks} as a linear system, where the unknowns are exactly the erased symbols.  Due to the sparsity of $\6H$, with large probability the linear system will include several univariate equations, \ie, containing only one erasure. Each of these equations can be solved to compute the corresponding unknown, which is then substituted into all the other equations.  This procedure is iterated until all the unknowns are found or, at some point, the linear system does not contain any univariate equation, \ie, all the unsolved equations contain at least two unknowns. 
In the former case we have a decoding success, while in the latter case we have a failure, due to a \textit{stopping set} \cite{Richurba}, \ie, a set of symbols participating to parity-check equations containing at least two unknowns each. If all the symbols forming a stopping set are erased, peeling decoding fails. The \emph{stopping ratio} $\beta$ of an LDPC code is defined as the minimum stopping set size divided by $n$. 

\subsection{Components of the SPAR protocol}\label{subsec:sparcompon}

 SPAR is based on a novel hash tree called Coded Merkle Tree (CMT), combined with an ad-hoc  \textit{hash-aware} peeling decoder. 

\paragraph*{\textbf{Coded Merkle Tree}} A CMT is a hash tree which is constructed from $\ell$ linear codes $\{\code^{(1)}, \cdots,\code^{(\ell)}\}$ over $\mathbb F_q$; the $i$-th code has length $n_i$ and dimension $k_i$.
Each code $\code^{(i)}$ is defined by the systematic generator matrix $\6G^{(i)}=[ \6I_{k_i} \hspace{0.5mm}|\hspace{0.5mm} \6A_i ]$, with $\6A_i\in\mathbb F_q^{k_i\times (n_i-k_i)}$ and $\6I_{k_i}$ being the identity matrix of size $k_i$. 
The CMT uses an integer $b$ which must be a divisor of all blocklength values $n_1,\cdots,n_\ell$.
Furthermore, one needs to have partitions for the sets $[1 , n_i]$, for $i \in [1 , \ell-1]$.
Namely, we have $\mathcal S_i = \left\{S_1^{(i)},\cdots,S_{k_{i+1}}^{(i)}\right\}$ which is a partition of $[1 , n_i]$, such that the $S_j^{(i)}$ are all disjoint and each one contains $b$ elements, since $k_{i+1}=n_i / b$.
Starting from $\6c\in\code^{(1)}$, we build the associated CMT $\cmt(\6c)$ as follows:
\begin{enumerate}
    \item set $i = 1$;
    \item for $j\in\{1,\cdots,k_{i+1}\}$, set $$u_j =  \concat\left(\hash\left(c^{(i)}_{z_1}\right),\cdots,\hash\left(c^{(i)}_{z_b}\right)\right),$$ with $\{z_1,\cdots,z_b\} = S^{(i)}_j$;
    \item encode $\6u = [u_1,\cdots,u_{k_{i+1}}]$ as $\6c = \6u\6G^{(i+1)}$;\footnote{Notice that, when LDPC codes are considered, encoding is conveniently performed using the parity-check matrix rather than the generator matrix. This implementation detail does not affect the conclusions of our  analysis but, considering encoding with the parity-check matrix, we would unnecessarily burden the notation. Therefore, we stick to encoding with the generator matrix.} 
    \item if $i<\ell-1$, increase $i$ and restart from step 2), otherwise set $\cmtroot(\6c) = \6u$.
\end{enumerate}

\paragraph*{\textbf{Hash-aware peeling decoder}}
A hash aware peeling decoder, described in \cite[Section 4.3]{yu2020coded}, is an algorithm that decodes a set of $\ell$ words which are expected to constitute a CMT.
Namely, let $\{\6x^{(1)},\cdots,\6x^{(\ell)}\}$, where $\6x^{(i)} \in \{\mathbb F_q\cup \epsilon\}^{n_i}$, be the words to be decoded. 
The hash-aware peeling decoder works in a top-down fashion and, at every iteration, uses the peeling decoder strategy (\ie, recover erasures that participate in univariate parity-check equations) for any layer of the CMT.
Additionally, the hash-aware peeling decoder verifies the consistency between symbols of connected layers of the tree via hash functions, whilst the symbols are recovered. 
Decoding fails whenever a stopping set or a failed parity-check equation is met, just like the conventional peeling decoder. Furthermore, the hash-aware peeling decoder fails in case check consistency fails for some layer. Finally, an undetected error is met (but not recognized by the decoder) if the decoded sequence is a codeword, but not the original one. 

\section{A general framework to study DAAs}
\label{sec:spar}
In this section we present a general framework to study DAAs, and then apply it to the SPAR protocol.
For brevity, we only give the fundamentals of the model; for further details concerning DAAs, we refer the interested reader to \cite{yu2020coded,al2021fraud}.

\subsection{A general model for DAAs}\label{subsec:highle}

We consider a game in which an adversary \adversary exchanges messages with $m$ players $\player_{1},\cdots,\player_{m}$, who cannot communicate one each other. 
Each player has access to an oracle \oracle, who can only perform polynomial time operations.
Every list of transactions is seen as a vector $\6u\in\mathbb F_q^k$.
We assume that the following information is publicly available:
\begin{enumerate}
\item[-] a validity function $f:\mathbb F_q^k\mapsto\{\sf{False}, \sf{True}\}$, which depends on the blockchain rules and on its current status;
\item[-]  two hash trees $\mathcal T, \mathcal T'$; 
\item[-] a $k$-dimensional code $\code\subseteq \mathbb F_q^n$ with generator matrix $\6G$.
\end{enumerate}
The game proceeds as follows:
\begin{enumerate}
\item \adversary chooses $\6u\in\mathbb F_q^k$ such that $f(\6u) = \sf{False}$ and $\tilde{\6c}\in\mathbb F_q^n$;
\item \adversary challenges the players with $(h_u, h_c)$, where $h_u = \mathcal T.{\sf{Root}}(\6u)$, $h_c = \mathcal T'.{\sf{Root}}(\tilde{\6c})$;
\item each player $\mathcal P_i$ selects $J_i\subseteq [1,n]$ with size $s$;
\item \adversary receives $U = \bigcup_{i = 1}^m J_i$;
\item to reply to a query containing the index $i$, \adversary must send $\{\tilde{c}_i, \mathcal T.{\sf{Proof}}(\tilde{\6c},i)\}$; \adversary is free to choose which queries to reply and which ones to neglect;
\item if a player does not receive a valid reply for any of his queries, then he discards $(h_u, h_c)$;
\item the players gossip all the valid answers to \oracle, which aims to produce a proof for one of the following facts:
\begin{enumerate}
\item[a) ] $\exists \tilde{\6c}\not\in\code$, such that $\mathcal T'.{\sf{Root}}(\tilde{\6c}) = h_c$;
\item[b) ] $\exists  \tilde{\6c}\in\code$ such that $\mathcal T'.{\sf{Root}}( \tilde{\6c}) = h_c$, $ \tilde{\6c} = \tilde{\6u}\6G$ and $\mathcal T.{\sf{Root}}(\6u)\neq h_u$;
\item[c) ] $\exists  \tilde{\6c}\in\code$ such that $\mathcal T'.{\sf{Root}}( \tilde{\6c}) = h_c$, $ \tilde{\6c} = \tilde{\6u}\6G$, $\mathcal T.{\sf{Root}}(\tilde{\6u}) = h_u$ and $f(\tilde{\6u}) = \sf{False}$. 
\end{enumerate}
\end{enumerate}

Let us also define two properties.

\begin{Def}
\textbf{Soundness:} if a player accepts $(h_u, h_c)$, then \oracle will be able to recover $\tilde{\6c}$ (and $\tilde{\6u}$) within a finite maximum delay.
\label{def:sound}
\end{Def}

\begin{Def}
\textbf{Agreement:} if a player accepts $(h_u, h_c)$, then all the other players will accept $(h_u, h_c)$ within a finite maximum delay.
\label{def:agree}
\end{Def}

Clearly, if \adversary wins the game, which happens with probability $\gamma$, soundness and agreement are caused to fail. We denote by $\gamma$ the Adversarial Success Probability (ASP), \ie, the probability that \adversary wins a random execution of the game.

It can be easily seen that, in our model, the players $\player_1,\cdots,\player_m$ correspond to the light nodes connected to a malicious node modeled by \adversary.
The oracle \oracle instead represents the fact that we assume any light node must be connected to at least one honest full node wishing to broadcast fraud proofs.
We remark that the hypotheses and properties that underlie our model are the same under which DAAs have been studied in the literature \cite{yu2020coded, al2021fraud, Dolececk2020, santini2022optimization}.
Finally, our model does not fix any hash tree, nor code family; thus, it can be used to study several blockchain networks.
We now proceed by describing how SPAR adapts to such a model, but it can be easily seen that also the protocol proposed in \cite{al2021fraud} fits into the model.

\subsection{DAAs in the SPAR protocol}\label{subsec:spadaa}

In SPAR, the CMT is instantiated using the code design procedure considered in \cite{Luby}, which produces an ensemble of LDPC codes whose parity-check matrices have at most column weight $v$ and at most row weight $w$.
As mentioned in Section \ref{subsec:sparcompon}, besides the CMT, SPAR requires the use of another hash tree, denoted by $\merkle$ and considered as a standard Merkle tree.

Let $\6u\in\mathbb F_q^k$ denote the list of transactions of a new block.
Then, a correctly constructed header contains $h_u = \mroot(\6u)$ and $h_c = \cmtroot(\6c)$, with $\6c = \6u\6G^{(1)}$.
However, in case of a DAA, the word $ \tilde{\6c}=\6c+\6e$ upon which $h_c$ is constructed may be any vector picked from $\mathbb F_q^n$. 
The authors of SPAR study the protection of the protocol against DAAs; namely, they initially consider the following two cases:
\begin{enumerate}
    \item[a) ] if $ \tilde{\6c}^{(i)}\not\in\code^{(i)}$, then the proof consists in sending the value of all the symbols that participate in a failed parity-check equation, except for one of them, together with their CMT proofs; we refer to such a proof as \textit{parity-check equation incorrect-coding proof};
    \item[b) ] if $\tilde{\6c}=\6c$ but $f(\6u) = {\sf{False}}$, the adversary succeeds only if the samples received by the oracle are not enough to allow the recovery of $\6u$ from $\tilde{\6c}$ through decoding.
\end{enumerate}
The scenario where the oracle finds a hash inconsistency is also considered, in which case \oracle can broadcast a fraud proof to the light nodes, called here \emph{hash inconsistency incorrect-coding proof}. 
 
The following bound for the ASP is derived \cite[Theorem 1]{yu2020coded}:
\begin{equation}
\label{eq:claimed_asp}
\gamma \leq \max\left\{(1-\alpha_{\rm min})^s\hspace{1mm},\hspace{1mm}2^{\max_{i}\{b(\alpha_i)n_i+ms\log(1-\alpha_i)\}}\right\}
\end{equation}
where $\alpha_i$ is the \emph{undecodable ratio} of $\code^{(i)}$, that is, the minimum fraction of coded symbols the adversary needs to make unavailable in order to prevent the oracle from full decoding, $\alpha_{\rm min} = \min_{i}(\alpha_i)$, and $s$ is the number of queries performed by each light node. Therefore, if the oracle is not able to decode due to the presence of a stopping set, the adversarial success probability computed in \cite{yu2020coded} is the probability that \emph{exactly} one player receives an answer to all its queries.

We argue here, instead, that a sufficient condition to break the soundness and agreement as defined in \cite{al2021fraud,yu2020coded}, and recalled in Section \ref{subsec:highle} is actually that \emph{at least} one player accepts a block which is invalid.

\begin{Prop}
In SPAR, an adversary cannot cause the soundness and agreement to fail with probability lower than 
\begin{equation}
\label{eq:refor_claimed_asp}
\gamma \leq \min\{1, \max\{1-(1-(1-\alpha_{\min})^s)^m\hspace{1mm},\hspace{1mm} t_2\}\},
\end{equation}
where $t_2=2^{\max_{i}\{b(\alpha_i)n_i+ms\log(1-\alpha_i)\}}$.
\end{Prop}
\begin{proof}
According to Definition \ref{def:sound}, the soundness fails if at least a player accepts the block header, but the oracle will not be able to dispatch a fraud proof. The probability that exactly one player accepts the challenge is lower than or equal to $(1-\alpha_{\min})^s$ and, therefore, the probability that exactly one player discards the challenge is larger than $1-(1-\alpha_{\min})^s$. Considering that there are $m$ players, the probability that all of them discard the block is larger than $\left[1-(1-\alpha_{\min})^s\right]^m$. So, finally, the probability that at least a player accepts the block is lower than
\[1-\left[1-(1-\alpha_{\min})^s\right]^m.\]
The rest of the proof is as in \cite[Theorem 1]{yu2020coded}.
\end{proof}



\section{Numerical examples}\label{sec:daa}

Let us consider the code parameters  proposed in \cite{yu2020coded} as a benchmark. It is shown in \cite[Table 2]{yu2020coded} that the most favourable value of the stopping ratio of the constructed ensemble ($\beta^*$)   is obtained when $w=8$ and the code rate is $R=1/4$, from which $v=6$ easily follows. As in \cite{yu2020coded} we consider two cases: a \emph{strong adversary} (SA) able to find stopping sets and erase the corresponding symbols, and a \emph{weak adversary} (WA) unable to find them and hence forced to erase random symbols. For the SA, the undecodable ratio is $\alpha^*=\beta^*=12.4\%$; in case of WA, we instead have $\alpha^*=47\%$\cite{yu2020coded}.
According to \cite[Table 2]{yu2020coded}, when $n=4096$, the probability that the code stopping ratio $\alpha$ is smaller than the ensemble stopping ratio is relatively small ($3.2\cdot10^{-4}$). 

In Table \ref{tab:gampers} we report the upper bound \eqref{eq:claimed_asp} and the newly assessed upper bound \eqref{eq:refor_claimed_asp} on the ASP, for some values of $s$, considering $n=4096$ and $m=1024$; notice that the new value is never smaller than the previously computed upper bound. Clearly, this may have sever security consequences.

\begin{table}[tbh!]
\centering
\caption{Values of \eqref{eq:claimed_asp} and \eqref{eq:refor_claimed_asp} for $m=1024$, $n=4096$.}
\begin{tabular}{|c||cc||cc|}
\hline
\multirow{2}{*}{$s$} & \multicolumn{2}{c||}{Upper bound on $\gamma$ \cite{yu2020coded}}                    &  \multicolumn{2}{c|}{New upper bound on $\gamma$} \\ \cline{2-5}         
                     & \multicolumn{1}{c|}{WA \eqref{eq:claimed_asp}}      & SA
                     \eqref{eq:claimed_asp}   & \multicolumn{1}{c|}{WA \eqref{eq:refor_claimed_asp}}      & SA \eqref{eq:refor_claimed_asp}         \\ \hline \hline
$8$                  & \multicolumn{1}{c|}{$6.23\cdot10^{-3}$}   & $\approx 1$                & \multicolumn{1}{c|}{$\approx 1$}   & $\approx 1$               \\ \hline
$35$                 & \multicolumn{1}{c|}{$2.24\cdot10^{-10}$} & $9.72\cdot10^{-3}$   & \multicolumn{1}{c|}{$2.29\cdot10^{-7}$}   & $\approx 1$               \\ \hline
$200$                & \multicolumn{1}{c|}{$\approx 0$}  & $3.17\cdot10^{-12}$  & \multicolumn{1}{c|}{$\approx 0$}   & $3.24\cdot10^{-9}$             \\ \hline
$2000$               & \multicolumn{1}{c|}{$\approx 0$}         &$\approx 0$   & \multicolumn{1}{c|}{$\approx 0$}   & $\approx 0$       \\ \hline
\end{tabular}
\label{tab:gampers}
\end{table}

 Conversely, once a target adversarial success probability is chosen, it is possible to compute a lower bound number of samples $s$ each player needs to ask for in order to stay below it, by inverting \eqref{eq:claimed_asp} and \eqref{eq:refor_claimed_asp}. Considering the same parameters as above ($n=4096$ and $m=1024$) we obtain the results in Table \ref{tab:s}.

\begin{table}[tbh!]
\centering
\caption{Values of $s$ (obtained by inverting \eqref{eq:claimed_asp} and \eqref{eq:refor_claimed_asp}) for $m=1024$, $n=4096$ and different values of $\gamma$.}
\begin{tabularx}{0.8\textwidth}{|Y||YY||YY|}
 \hline
\multirow{2}{*}{$\gamma$} & \multicolumn{2}{c||}{Lower bound on $s$ \cite{yu2020coded}}                    &  \multicolumn{2}{c|}{New lower bound on $s$} \\ \cline{2-5}         
                     & WA      & SA   & WA      & SA        \\ \hline \hline
$10^{-2}$                 & \multicolumn{1}{c}{$8$} & $35$   & \multicolumn{1}{c}{$19$}   & $88$               \\ \hline
$10^{-5}$                & \multicolumn{1}{c}{$19$}  & $87$  & \multicolumn{1}{c}{$30$}   & $140$             \\ \hline
$10^{-10}$             & \multicolumn{1}{c}{$37$}         & $174$   & \multicolumn{1}{c}{$48$}   & $227$       \\ \hline
\end{tabularx}
\label{tab:s}
\end{table}

 We notice that the actual number of samples asked by each node is much larger than expected, resulting in a larger sampling cost $S$, which increases linearly with $s$ as follows \cite{yu2020coded}
 \[
 S=s\left(\frac{B}{k}+\left[y(b-1)+yb(1-R)\right]\log_{bR}\frac{k}{Rt}\right),
 \]
 where $B$ is the block size, $y$ is the hash size and $b$ is the number of batched hashes in each layer. The header size is $H=t\ell_{\mathcal{H}}$, where $\ell_{\mathcal{H}}=256$ is the binary length of the digests.

 We assess the sampling cost $S$, normalized with respect to the block size $B$, considering $m=1024$, $R=1/4$, $k=1024$ symbols, $B=1\hspace{1mm}\mathrm{MB}$, $b=8$ and $t=256$ hashes and some different values of the ASP $\gamma$, in Table \ref{tab:samplingcost}. A comparison with the optimized ASBK protocol \cite{santini2022optimization} is also reported, for which we have considered the same block size, and codes defined over a field of size $2^{256}$. As expected, the optimized ASBK protocol results in smaller sampling costs than the SPAR protocol (this also held true for the original ASBK protocol \cite[Fig. 4]{yu2020coded}.)

\begin{table}[tbh!]
    \centering
    \caption{Sampling cost $S$ normalized to the block size $B$ for $m=1024$, $n=4096$ and different values of $\gamma$.}
    \begin{tabular}{| E | | G G | | G G | | D |}
    \hline
    \multirow{2}{*}{$\gamma$}&\multicolumn{2}{c||}{Lower bound on $S/B$ \cite{yu2020coded}}&\multicolumn{2}{c||}{Lower bound on $S/B$ \cite{yu2020coded}}&{Lower bound on $S/B$ \cite{santini2022optimization}}\\\cline{2-6} &WA&SA&WA&SA&-  \\ \hline \hline
$10^{-2}$                 & $0.0233$ &  {$0.1019$}   & $0.0553$   &  {$0.2563$}     &   {$0.0278$}        \\ \hline
$10^{-5}$                & {$0.0533$}  &  {$0.2534$}  & {$0.0874$}   &  {$0.4077$}     &       {$0.0358$} \\ \hline
$10^{-10}$             & $0.1078$         &  {$0.5068$}   & {$0.1398$}   &  {$0.6611$}   &   {$0.0435$}  \\ \hline
    
    \end{tabular}
\label{tab:samplingcost}
\end{table}

However, it should be noticed that SPAR has the advantage of relying on a fixed header size whereas in ASBK the header size increases as the square root of the block size. Therefore, considering the same setting, we have compared the total amount of downloaded data $D$ (sampling cost plus header size) using SPAR, to that obtained using the optimized ASBK protocol in Tables \ref{tab:totalcost1}, \ref{tab:totalcost2} and \ref{tab:totalcost3}, where we have also reported the header size $H$ for the optimized ASBK protocol, when $B=1$ MB, $B=10$ MB and $B=100$ MB, respectively. The header size for SPAR does not depend on the block size and its value is $t\ell_{\mathcal{H}}=8.192$ kB. Notice that this amount of data must be downloaded by any light node during the regular course of the protocol, independently of the malicious behaviour of some full nodes, possibly resulting in the additional download of fraud proofs.
 
\begin{table}[tbh!]
    \centering
    \caption{Total amount of downloaded data normalized to the block size $B=1$ MB for $m=1024$ and different values of $\gamma$.}
    \begin{tabular}{| E | | C C | | D | | F |}
    \hline
    \multirow{2}{*}{$\gamma$}&\multicolumn{2}{c||}{New lower bound on $D/B$}&{Lower bound on $D/B$ \cite{santini2022optimization}}&$H$ [kB] \cite{santini2022optimization}\\\cline{2-5} &WA&SA&-&-  \\ \hline \hline
$10^{-2}$                 &  {$0.0635$} & $0.2645$   &              $0.0454$ & \multirow{3}{*}{$20.411$} \\ \cline{1-4} 
$10^{-5}$                &  {$0.0956$}  &  $0.4159$ &   $0.0544$     &       \\ \cline{1-4}
$10^{-10}$             &  {$0.148$}         &  $0.6693$  &  $0.0639$  &\\ \hline

    \end{tabular}
\label{tab:totalcost1}
\end{table}

\begin{table}[tbh!]
    \centering
    \caption{Total amount of downloaded data normalized to the block size $B=10$ MB for $m=1024$, $n=4096$ and different values of $\gamma$.}
    \begin{tabular}{| E | | C C | | D | | F |}
    \hline
    \multirow{2}{*}{$\gamma$}&\multicolumn{2}{c||}{New lower bound on $D/B$}&{Lower bound on $D/B$ \cite{santini2022optimization}}&$H$ [kB] \cite{santini2022optimization}\\\cline{2-5} &WA&SA&-&-  \\ \hline \hline
$10^{-2}$                 &  {$0.009$} & $0.0386$   &              $0.0099$ & \multirow{3}{*}{$52.244$} \\ \cline{1-4} 
$10^{-5}$                &  {$0.0137$}  &  $0.0609$ &   $0.0123$     &       \\ \cline{1-4}
$10^{-10}$             &  {$0.0214$}         &  $0.0983$  &  $0.0154$  &\\ \hline
    \end{tabular}
\label{tab:totalcost2}
\end{table}

\begin{table}[tbh!]
    \centering
    \caption{Total amount of downloaded data normalized to the block size $B=100$ MB for $m=1024$, $n=4096$ and different values of $\gamma$.}
    \begin{tabular}{| E | | C C | | D | | F |}
    \hline
    \multirow{2}{*}{$\gamma$}&\multicolumn{2}{c||}{New lower bound on $D/B$}&{Lower bound on $D/B$ \cite{santini2022optimization}}&$H$ [kB] \cite{santini2022optimization}\\\cline{2-5} &WA&SA&-&-  \\ \hline \hline
$10^{-2}$                 &  {$0.0012$} & $0.0051$   &              $0.0022$ & \multirow{3}{*}{$158.03$} \\ \cline{1-4} 
$10^{-5}$                &  {$0.0018$}  &  $0.008$ &   $0.0025$     &       \\ \cline{1-4}
$10^{-10}$             &  {$0.0028$}         &  $0.013$  &  $0.0031$  &\\ \hline
    \end{tabular}
\label{tab:totalcost3}
\end{table}

We observe that, for relatively small and moderate values of the block size, despite the larger header size, the use of the ASBK protocol is preferable even if a weak adversary is taken into account. Instead, when the block size is large, SPAR is very convenient in the presence of a weak adversary, but still more costly than ASBK if the adversary is strong.

\section{Conclusion}\label{sec:concl}

By carefully analyzing the SPAR protocol  we have shown that the actual sampling cost required by the scheme, in order to achieve target security guarantees, is much larger than that initially expected. Moreover, it is shown that, in many practical scenarios, the quantity of data light nodes have to download is larger than that of other well-known schemes. 

\bibliographystyle{IEEEtran}
\bibliography{References}

\end{document}